\documentclass[reprint,nofootinbib,amsmath,amssymb,aps,pra]{revtex4-2}

\usepackage{graphicx}
\usepackage{dcolumn}
\usepackage{bm}
\usepackage[utf8]{inputenc}
\usepackage[T1]{fontenc}

\usepackage{etoolbox}

\usepackage{blindtext}
\usepackage[polish, english]{babel}
\usepackage{makecell}
\usepackage{multirow}
\usepackage{amsfonts}
\usepackage{anyfontsize}
\usepackage{amsmath}
\usepackage{amsthm}
\usepackage{tikz}
\usepackage{extpfeil}
\usepackage{tikz-cd}
\usepackage{physics}
\usepackage{caption}
\usepackage{subcaption}
\usepackage{xcolor}
\usepackage{dsfont}
\usepackage{appendix}
\usepackage[breaklinks=true]{hyperref}
\hypersetup{
    colorlinks=true,
    linkcolor=blue,
    filecolor=blue,
    urlcolor=blue,
    citecolor=blue
}

\usepackage[polish, english]{babel}

\usepackage{graphicx,makecell}%
\usepackage{multirow}%
\usepackage{amsfonts}%
\usepackage{amsthm}%
\usepackage{mathrsfs}%
\usepackage{xcolor}%
\usepackage{tikz}
\usepackage{extpfeil}
\usepackage{tikz-cd}
\usepackage{physics}
\usepackage{dsfont}
\usepackage[mathscr]{euscript}
\usepackage{scalerel}
\usepackage{tikz}
\usetikzlibrary{svg.path}
\usepackage{stackengine}
\newcommand\stackrqarrow[2]{%
    \mathrel{\stackunder[2pt]{\stackon[4pt]{$\rightsquigarrow$}{$\scriptscriptstyle#1$}}{%
            $\scriptscriptstyle#2$}}}

\newextarrow{\xbigtoto}{{20}{20}{20}{20}}    {\bigRelbar\bigRelbar{\bigtwoarrowsleft\rightarrow\rightarrow}}
\newextarrow{\xbigto}{{20}{20}}
   {\bigRelbar{\bigtwoarrowsleft\rightarrow}}

\newtheorem{theorem}{Theorem}
\newtheorem{proposition}[theorem]{Proposition}
\newtheorem{definition}{Definition}%
\newtheorem{corollary}[theorem]{Corollary}

\newtheorem{result}[theorem]{Result}
\newtheorem{axiom}{Axiom}

\newcommand*\interior[1]{\mathring{#1}}
\newcommand*\closure[1]{\overline{#1}}

\begin{document}

\preprint{APS/123-QED}

\title{Wigner and friends, a map is not the territory!\\ Contextuality in multi-agent paradoxes}
\author{Sidiney B. Montanhano}
\email{s226010@dac.unicamp.br}
\affiliation{Instituto de Matemática, Estatística e Computação Científica, Universidade Estadual de Campinas (Unicamp), 13083-859, Campinas, São Paulo, Brazil}
\date{\today}

\begin{abstract}
Multi-agent scenarios, like Wigner's friend and Frauchiger-Renner scenarios, can show contradictory results when a non-classical formalism must deal with the knowledge between agents. Such paradoxes are described with multi-modal logic as violations of the structure in classical logic. Even if knowledge is treated in a relational way with the concept of trust, contradictory results can still be found in multi-agent scenarios. Contextuality deals with global inconsistencies in empirical models defined on measurement scenarios even when there is local consistency. In the present work, we take a step further to treat the scenarios in full relational language by using knowledge operators, thus showing that trust is equivalent to the Truth Axiom in these cases. A translation of measurement scenarios into multi-agent scenarios by using the topological semantics of multi-modal logic is constructed, demonstrating that logical contextuality can be understood as the violation of soundness by supposing mutual knowledge. To address the contradictions, assuming distributed knowledge is considered, which eliminates such violations but at the cost of lambda-dependence. We conclude by translating the main examples of multi-agent scenarios to their empirical model representation, contextuality is identified as the cause of their contradictory results.
\end{abstract}

\maketitle
\tableofcontents

\section{Introduction}

Multi-agent paradoxes \cite{Frauchiger_2018,Nurgalieva_2020,Vilasini_2019} are violations of agreement among agents about some global information. They appear in generalizations of the Wigner's friend (thought experimental) scenario \cite{Wigner1995,Deutsch1985}, which itself extends Schrödinger's famous thought experiment with his cat \cite{Schrodinger1935}. The exploration of such paradoxes and other related phenomena has proven to be significant for the fundamental understanding of quantum theory and its interpretations \cite{brukner2015quantum,e20050350,Bong2020,e23080925,Haddara_2023,Rela_o_2020,Rossi_2021,AllardGuerin2021,schmid2023review}. The formal construction of these paradoxes employs the language of modal logic \cite{sep-logic-modal,sep-phil-multimodallogic,sep-logic-epistemic} to show how (thought experimental) scenarios, found both in quantum theory and in other non-classical theories beyond quantum theory, presents a violation in the structure of classical logic.

Contextuality in its standard definition \cite{Kochen1975} deals with global inconsistencies in measurements even when there is local consistency, that is, even if the model is non-disturbing, generalizing the famous phenomena of nonlocality \cite{PhysicsPhysiqueFizika.1.195} and the condition of no-signaling between observers \cite{Abramsky_2011}. This concept has been formalized in various ways in the literature, including the topological \cite{Abramsky_2011,Abramsky_2012,abramsky_et_al:LIPIcs:2015:5416,okay2017topological,2020Okay,montanhano2021contextuality,montanhano2021characterization}, the algebraic \cite{BirkhoffvonNeumann,PhysicsPhysiqueFizika.1.195,Kochen1975,Gleason}, the geometrical \cite{Cabello_2014,amaral2017geometrical} and generalizations \cite{Dzhafarov2015ContextualitybyDefaultAB,1994SORKIN,Spekkens2008,Schmid2021,https://doi.org/10.48550/arxiv.2202.08719}. This phenomenon has a significant impact on the fundamental properties of quantum theory \cite{doring2020contextuality}, and is necessary for any potential computational advantage over classical computers \cite{shahandeh2021quantum,Howard2014}. Its formal description employs the categorical language of sheaves and presheaves \cite{Abramsky2011}, enabling the construction of bundle diagrams for each model \cite{Beer_2018}. Equivalently, a contextual model implies the inability to describe it classically, in the sense of an embedding of propositions into a set of Boolean propositions with probabilistic valuation satisfying Kolmogorov's probability axioms, even with additional propositions in the set of Boolean propositions serving as hidden variables.

Our objective in this paper is to construct a map between empirical model and multi-agent scenarios such that the contextuality of the first and the multi-agent paradox of the second could be identified as the same phenomenon. To achieve this objective, we will need to rewrite the Sheaf Approach \cite{Abramsky_2011}, the formalism of contextuality closest to multi-modal logic, in such a way that we can understand how each part of the construction of empirical models can be identified in a multi-agent scenario. This will require a more refined treatment of multi-modal logic, as concepts like fundamental truth become elusive when non-classicality is present. Such refinement, which explicitly utilizes knowledge operators and the trust relation defined between sets and agents, can be understood as the logical formalization of Alfred Korzybski's statement ``A map is not the territory'' \cite{korzybski1933science} applied to the knowledge that agents can access. Consequently, our main finding is that, contrary to what is claimed in the literature \cite{Nurgalieva_2019}, modal logic is suitable for quantum and other non-classical settings. To illustrate this, we analyze three famous examples that accept the inverse in the constructed map.

The paper is structured as follows. We employ the topological semantics of multi-modal logic (specifically the $S4$ system) to initially investigate the application of trust \cite{Nurgalieva_2019,Vilasini_2019} when the knowledge operators are explicitly utilized. This exploration is grounded in the idea that knowledge is inherently relational—something is not merely known; it must be known by someone. Trust can be understood as a relational way to define the Truth Axiom; in fact, they are equivalent when seen by the topology induced by distributed knowledge of the agents, as shown in section \ref{Knowledge and trust}. We can thus use the knowledge operators and trust to create a translation between multi-agent scenarios and empirical models up to restrictions. In Section \ref{Contextuality}, we systematically identify the elements of an empirical model as elements of a multi-agent scenario in multi-modal logic, exploring the implications of this identification, and discussing the limitations of empirical models in handling generic multi-agent scenarios. The violation of soundness described in Ref. \cite{Nurgalieva_2019} that appears as the failure of classical logic to deal with quantum theory is identified as the hidden imposition of mutual knowledge on the agents. This implies the conclusion that modal logic fails to deal with multi-agent paradoxes. Interestingly, this issue disappears when distributed knowledge is imposed. A similar resolution occurs when translating the contextuality conditions to multi-modal language, albeit with the cost of lambda-dependence. Next, in section \ref{Multi-agent scenarios}, we work out the three examples of multi-agent paradoxes: the Wigner's friend scenario, the Frauchiger-Renner scenario, and the Vilasini-Nurgalieva-del Rio scenario, in topological semantics and translate them to their sheaf representation. We then identify contextuality as the origin of their paradoxes when they appear. In Section \ref{Discussion}, we provide insights into the results and explore future research possibilities. The appendices serve to fix the notation used and give the basics of modal logic in appendix \ref{Modal Logic}, and the sheaf approach to contextuality in appendix \ref{Sheaf approach}.

\section{Knowledge, trust, and the construction of fundamental truth}
\label{Knowledge and trust}

In this section, we will address the concept of truth in modal logic and its relationship with different types of knowledge operators. The Truth Axiom $\mathbf{T}$, which states that if a proposition is known by an agent, then it is true, appears to be too strong to handle multi-agent scenarios. Hence, as is conventional in the literature of multi-agent scenarios, we will not directly assume the Truth Axiom $\mathbf{T}$. Instead, we will use trust between agents, as suggested by \cite{Nurgalieva_2019}, to weaken the classical notion of fundamental truth. Even with this more relational approach to multi-agent scenarios, the same paradoxes arise. The reason for this, as we will see, is that we can use the more refined knowledge operator in the hierarchy of operators, distributed knowledge, to reconstruct a fundamental truth, which would be equivalent to the Truth Axiom.

We will work with the system $\mathbf{S4}$, which imposes Kripke semantics (Distribution Axiom) of possible worlds and accessibility relations, the transitivity (Positive Introspection Axiom) and the reflexive (Truth Axiom) properties of the accessibility relations, standard when dealing with epistemic logic \cite{sep-logic-epistemic}. In particular we will mostly work with the topological semantics of the system $\mathbf{S4}$, naturally related to a topological view of the Sheaf Approach. In such semantics, we can think of the propositions of the system $\mathbf{S4}$ as open sets in a set of possible worlds, with the usual operations between open sets (union, intersection, complement, interior, closure, complement) playing the role of logical operations between propositions (or, and, negation, necessity, possibility). For an introduction to the logical content of what follows, see appendix \ref{Modal Logic}.

\subsection{Trust instead of fundamental truth}

Let's introduce and justify the concept of trust between agents as an alternative to a fundamental truth. Multi-agent scenarios defined with modal logic in the literature use only trust between agents. To construct the map we aim for between empirical models and multi-agent scenarios in the next section, we will define here a generalization of the concept of trust that also applies to sets of agents. With such trust relation, we will formally define multi-agent scenarios.

Knowledge, mutual knowledge, and distributed knowledge operators (see appendix subsection \ref{Knowledge}) are important for writing formulas when one imposes the following principle: there is no knowledge without an agent. This principle can be understood as the embodiment of the obvious idea that fundamental truth, in the sense of being absolute to all agents, following Plato's ``justified true belief'' definition of knowledge, is a philosophical position rather than an empirical fact. One can only assume something is true for all agents, but cannot test such a thing.

As we will see in the examples in Section \ref{Multi-agent scenarios}, quantum theory imposes limitations on this Platonic view of knowledge aforementioned. One way to attempt to circumvent such limitations is to weaken the concept of knowledge to the ``justified belief'' definition, which does not presuppose any fundamental truth, but rather relies on justification based on the inevitably incomplete data accessible to the agent. Therefore, every formula must be evaluated through a knowledge operator. Axioms $\mathbf{K}$ and $\mathbf{4}$ do not present any issues once the operator is introduced. However, $\mathbf{T}$ relies on the notion of fundamental truth, leading to certain philosophical complications. Let's ignore them by allowing beliefs to be on the same level as knowledge. To simplify matters, we will assume any further mechanism beyond the scope of this paper to distinguish them. Given the absence of an absolute notion of knowledge, we must find knowledge by trust between agents \cite{Vilasini_2019,Nurgalieva_2019}.

\begin{axiom}[Trust]
The trust relation $\rightsquigarrow$ between agents $i$ and $j$ is given by
\begin{equation}
(j\rightsquigarrow i)\leftrightarrow(K_i K_j\phi\to K_i\phi)\forall\phi,
\end{equation}
meaning ``$i$ trusts $j$''.
\label{Trust Original}
\end{axiom}

This notion of trust between agents must be generalized to deal with sets of agents, as it will become important when we discuss the relationship between the trust relation and contexts. An agent $i$ could not trust the agents of a set $G$ separately, but only when seen as a collective entity. In other words, the agent $i$ trusts the distributed knowledge of $G$. In this sense, ``$i$ trusts $G$'' if and only if, for all propositions, the knowledge of $i$ that the distributed knowledge of $G$ implies the knowledge of $i$. This is the weakest way to describe such a relation, where all agents in $G$ could not know $\phi$ individually\footnote{A simple example of applying this concept of trust arises in a presidential election within a presidential system with direct voting. Agent $i$ wishes to determine which candidate has been elected. Immediately before the election results are disclosed and the winner becomes mutual knowledge, agent $i$ cannot rely solely on individual agents to ascertain the victor, as knowledge is distributed across the electorate. It is agent $i$'s trust in the electorate as a whole, denoted as $G$, that leads them to accept the election outcome.}. Since for an agent $i$, we have that $K_i$, $D_i$, and $E_i$ are equivalent, we can generalize agent $i$ to a set $G'$ of agents with the mutual knowledge operator $E_{G'}$ in the previous argument with minimal modifications, or even restrict trust in $G$ to the mutual knowledge of its agents. The two trust relations can thus be defined as follows.

\begin{axiom}[Trust (sets of agents)]
The trust relation $\stackrqarrow{D}{}$ between a set of agents $G'$ and a set of agents $G$ is given by
\begin{equation}
(G\stackrqarrow{D}{} G')\leftrightarrow(E_{G'}D_G\phi\to E_{G'}\phi)\forall\phi,
\end{equation}
meaning ``$G'$ trusts $G$'', and the trust relation $\stackrqarrow{E}{}$
\begin{equation}
(G\stackrqarrow{E}{} G')\leftrightarrow(E_{G'}E_G\phi\to E_{G'}\phi)\forall\phi,
\end{equation}
meaning ``$G'$ trusts the consensus of $G$''.
\label{Trust sets}
\end{axiom}

We can formally define multi-agent scenarios through Kripke semantics, knowledge operators, and the generalization of the concept of trust relation.

\begin{definition}
    A multi-agent scenario is given by a tuple 
    \begin{equation}
    (\Sigma, I, \{K_i\}_{i\in I}, \{E_G\}_{G\in\mathcal{P}(I)}, \{D_G\}_{G\in\mathcal{P}(I)}, \stackrqarrow{D}{}, \stackrqarrow{E}{}),
    \end{equation}
    where each element $i\in I$ is an agent, $\Sigma$ is the set of possible worlds the agents explore, $\{K_i\}_{i\in I}$ are the knowledge operators of each agent, $\{E_G\}_{G\in\mathcal{P}(I)}$ and $\{D_G\}_{G\in\mathcal{P}(I)}$ are the operators of mutual and distributed knowledge among sets of agents, $\stackrqarrow{E}{}$ and $\stackrqarrow{D}{})$ are the trust relations between sets of agents.
\end{definition}

\subsection{The relation between the topology of different kinds of knowledge}

The topological semantics of the system $\mathbf{S4}$ is deeply related to knowledge. The definition of the knowledge operator $K$ in Kripke semantics can be rewritten as:
\begin{equation}
(M,w\vDash K\phi)\leftrightarrow(M,U^w\vDash \phi)
\end{equation}
In other words, in the world $w$, the agent knows something if and only if for all worlds in the element $U^w$ of the topological basis of the Alexandrov topology, that something is true. Here again, we have the problem of fundamental truth, with the important property by $\mathbf{T}$ that $w\in U^w$, which allows one to interpret $U^w$ as the natural neighborhood of $w$. In this sense, an agent knows something in a world if it is true in a neighborhood of such a world.

Epistemic logic with more than one agent defines an Alexandrov topology for each accessibility relation, which can be interpreted as different ways the agents perceive the worlds. We have the following relationship:
\begin{equation}
(K_i\phi\to K_j\phi)\forall\phi\leftrightarrow (R_{j}\subseteq R_{i})\leftrightarrow(\tau_j\subseteq\tau_i)
\end{equation}
between the knowledge operators, the induced relation, and the topology, respectively, in Kripke and topological semantics. In particular, one can show that the relationships 
\begin{equation}
    (K_i\phi\to D_I\phi)\forall\phi\leftrightarrow (R_{D_I}\subseteq R_{i})\leftrightarrow(\tau_{D_I}\subseteq\tau_i)
    \label{individual implies distributed}
\end{equation}
and
\begin{equation}
    (E_I\phi\to K_i\phi)\forall\phi\leftrightarrow (R_{i}\subseteq R_{E_I})\leftrightarrow(\tau_{i}\subseteq\tau_{E_I})
    \label{mutual implies individual}
\end{equation}
hold. They state that something known to an agent $i\in I$ is also known distributively, and that something mutually known is known to any agent in $i\in I$, respectively.  Additionally, one can show that a fundamental property of distributed knowledge is 
\begin{equation}
    (D_I\phi\to\phi)\forall\phi,
    \label{distributed implies fundamental}
\end{equation}
meaning the distributed knowledge of something implies its truth, which follows from the Truth Axiom $\mathbf{T}$. 

The hierarchy of knowledges presented in the previous paragraph (mutual knowledge implies individual knowledge from \ref{mutual implies individual}; individual knowledge implies distributed knowledge from \ref{individual implies distributed}; distributed knowledge implies fundamental truth from \ref{distributed implies fundamental}) will be important in which follows. Each relationship between the knowledge operators, usually represented by its topological incarnation in topological semantics, will be explored in the translation of an empirical model. But first, as said before, we need to address fundamental truth.

\subsection{Recovering fundamental truth from trust}
\label{2.3}

By definition, the finest topology generated by all agents is $\tau_{D_I}$, the topology given by the distributed knowledge operator of the set of all agents, which implies that the most refined knowledge that this set of agents can construct is given by the distributed knowledge. Therefore, there is no way to construct any other knowledge operator that captures more knowledge of fundamental truth than $D_I$. Once we need to understand knowledge in a relational manner, always explicitly specifying the knowledge operator, and both $(D_I\phi\to\phi)\forall\phi$ and $(E_I\phi\to K_i\phi)\forall\phi$ hold, we obtain the following proposition, the proof of which is given in appendix \ref{proof}.

\begin{proposition}
The following statements are valid:
\begin{itemize}
    \item Axiom $\mathbf{T}$ turns trust relations vacuous.
    \item The trust relation $\stackrqarrow{D}{}$, along with the condition that $(\phi\to D_I\phi)\forall\phi$, induces a fundamental truth.
\end{itemize}
 \label{Prop}
\end{proposition}

What the above proposition intuitively means is that one can say there is a limit to the knowledge a set of agents can access, which is the distributed knowledge, and there is no way to distinguish this limit from a fundamental limit of reality\footnote{In this sense, it is not surprising that an isolated population in an approximately stable environment that comes into contact with another civilization can suffer a significant impact on their culture. If they survive, their distributed knowledge usually loses its fundamentality.}\footnote{One could use a Bayesian vision to justify the existence of absolute truth through an induction argument, which holds in a classical description of reality, but it must be limited by Kant's epistemology.}. 

\section{The modal logic of an empirical model and contextuality}
\label{Contextuality}

In this section, we will construct the map that turns a generic empirical model into a multi-agent scenario in such a way that we can identify the paradoxes of the former with the contextuality of the latter. For simplicity, we will deal with finite objects. 

As we will see, assumptions about the nature of a scenario, such as classicality manifested in the form of global concordance between agents, impose restrictions on the possible worlds presumed to be accessible to agents through their topologies. We will demonstrate that paradoxes arise when an agent discovers that these worlds are insufficient to explain the scenario, potentially leading to the mistaken conclusion that modal logic is incapable of handling non-classical cases. We adopted a strategy of explicitly defining the knowledge operators, with a particular focus on identifying limits on the knowledge of a set of agents. This approach leads to another interesting consequence—the possibility of rethinking how possible worlds are chosen. We will use this when dealing with events in an empirical model. 

Another important point is that we need to make it clear that the agents' knowledge will not be distorted when transferred between agents. In order for the trust relationship to truly be trustworthy, we need to define what makes someone trustworthy, even if there is trust.

\begin{definition}
Given that $j\rightsquigarrow i$, an agent $j$ is trustworthy to the agent $i$ if $(K_i K_j\phi\to K_j\phi)\forall\phi$, i.e. if any information that agent $i$ knows from agent $j$ must also be known by agent $j$.
\label{trustworthy}
\end{definition} 

This condition is essential to prevent any hidden information; therefore, trust implies that the topology of the trusting party is finer than the trustworthy part. With trust, an agent can reconstruct all the information provided by its trustworthy part, which encompasses all of its information. An agent that is terminal in the network generated by the trust relation can, under this condition, reconstruct the information of all the agents, obtaining a global perspective of the knowledge.

\subsection{A measurement to each agent}

To construct the map we seek, we will follow the steps involved in building an empirical model. First, we start with a set $X$ containing elements referred to as measurements. Subsequently, we define a cover of contexts $\mathcal{M} \subset \mathcal{P}(X)$ that satisfies $\bigcup \mathcal{M} = X$, and $C' \subset C \in \mathcal{M}$ implies $C' \in \mathcal{M}$. Next, we apply a sheaf of events $\mathcal{E}:\left<X,\mathcal{M}\right>^{op}\to\textbf{Set}$, associating outcomes $O^{U}$ with each context $U \in \mathcal{U}$ — the local events — thus defining a measurement scenario. The global events follow from the sheaf property of $\mathcal{E}$. An empirical model is then characterized by a presheaf $\mathcal{D}_R:\textbf{Set}\to\textbf{Set}::O^{U}\mapsto\left\{\mu^{O^{U}}_R\right\}$, with $R$ a semiring, typically of the probabilistic or boolean type. Such a presheaf specifies distributions with values in $R$ for the outcomes of each context, usually imposing the no-disturbance condition $\mu^{O^{j}}_R|_{kj}=\mu^{O^{k}}_R|_{kj}$. Finally, contextuality arises when it becomes impossible to explain these distributions as marginalizations of a distribution on the global events. For further details, refer to Appendix \ref{Sheaf approach}.

The natural mapping of the measurements of an empirical model to a multi-agent scenario is achieved through the agents' measurements, as each agent in a multi-agent scenario is restricted to a measurement. Our first identification is as follows:
\begin{itemize}
    \item The measurements of an empirical model are mapped to the agents of its corresponding multi-agent scenario.
\end{itemize}
It's important to note that what we are identifying here is each measurement of an empirical model with a measurement of an agent in a multi-agent scenario. This constraint differs from scenarios in which agents have multiple measurements and the free will to choose among them, as in standard nonlocality scenarios. Through distributed knowledge, we have:
\begin{equation}
(K_i (K_i \phi)\to D_I (K_i \phi)\to K_i \phi)\forall\phi,
\end{equation}
which means that an agent trusts itself. This is in contrast to scenarios where an agent can choose between incompatible measurements. Therefore, agents in a multi-agent scenario cannot choose their measurements. If they could, each measurement would define different agents who cannot trust each other due to the incompatibility of their measurements.

\subsection{Contexts comes from trust}

Contexts can be understood as an island of classicality at the measurement level. Every measurement is nothing more than the marginalization of the ``mother'' measurement of the entire context. Such a definition of a context allows the construction of stochastic maps between all subcontexts, and these maps can define the probabilities of the context given the marginals.

Once the measurements are identified with the agents in the multi-agent scenario, the covering of contexts locally define agents who are in a classical environment, with the maps transmitting knowledge from one subcontext to another. Since the subcontexts, and therefore the agents, are terminal in the network of maps between subcontexts, and by imposing that they are all trustworthy, a condition called ``flasque beneath the cover'' in the literature of the Sheaf Approach, they all have a global view of the context $U$ in which they are. In logical notation, for any subcontexts $G,G'\in U$ we have the following:
\begin{equation}
    (E_G E_{G'}\phi\to E_G\phi)\forall\phi.
\end{equation}

Thus, in a context its classicality imposes that for any two subcontexts $G,G'\in U$ the trust relation $G\stackrqarrow{E}{} G'$ holds, with $\stackrqarrow{E}{}$ being an equivalence relation between the subcontexts of $U$ (reflexivity follows from trust itself, symmetry follows from the symmetry in the choice of subcontexts of $U$, and transitivity follows from the transitivity of maps between subcontexts), identified with the stochastic maps. Our second identification is as follows:
\begin{itemize}
\item Contexts are sets of agents in which the trust relation $\stackrqarrow{E}{}$ is an equivalence relation between subcontexts.
\end{itemize}
With these two identifications, we can rewrite the hypergraph of compatibility $\left<X, \mathcal{U}\right>$ of any empirical model as parts of a multi-agent scenario with the property of being covered by sets of agents with an equivalence trust relation $\stackrqarrow{E}{}$ between their subsets.

\subsection{Global events follows from the topology of mutual knowledge}

The next step is to identify the events of the empirical model given by the Sheaf of events. As elements of reality to which we have empirical access, such events must naturally be associated with some structure involving the possible worlds. However, worlds are defined globally, while events do not need to be. Here, we will see that such a distinction is related to the topology we are dealing with.

To overcome this apparent obstacle, the strategy is to use pointless topology \cite{Johnstone1983}. In this formalism, we start with the topology and its elements — the open sets — as the primitives of the topological space. In the topological semantics of modal logic, the use of pointless topology implies that we will take propositions, represented by the open sets, as the primitives, which is equivalent to the standard formalism where possible worlds are the primitive objects. When propositions are considered as primitives, the focus shifts to the operator of knowledge, its topology, and the set of agents that define it, with possible worlds emerging as consequent constructions derived from the propositions.

One cannot know the fundamental possible worlds, even if they exist, but only the propositions one can access. In other words, the worlds are defined by the propositions, and not the other way around. Therefore, possible worlds must be defined by the topology to which an agent has access, as the elements of a basis of such topology. Even with the most refined set of propositions, the questions that agents pose about the world cannot be taken as refined as the fundamental truth. Thus, any construction must begin with the limitation of the agents in ``mapping the territory.'' In an empirical model, this situation becomes more complicated, as we also have the local events organized into contexts as the propositions attempting to find out the possible world in which the scenario is in a locally classical way. Generally, there isn't a ``global map'' of the ``territory'' when dealing with empirical models.

Let $\left<X,\mathcal{U}\right>$ be the hypergraph of compatibility of a connected measurement scenario. Since it is connected and $\stackrqarrow{E}{}$ satisfies the conditions that determine a context, every element of $\mathcal{U}$ is a terminal object in the trust between sets of agents. In logical terms, for any $G\in \mathcal{P}(I)$ representing a element of $\mathcal{U}$, with $\mathcal{P}(I)$ being the power set of $I$, we have $(E_G\phi\to E_I\phi)\forall\phi$, implying $\tau_{E_G}=\tau_{E_I}$ for all $G$ representing a context, as $G\subseteq I$ implies $\tau_{E_I}\subseteq\tau_{E_G}$ and $(E_G\phi\to E_I\phi)\forall\phi$ implies $\tau_{E_G}\subseteq\tau_{E_I}$. Let's call $\mathcal{B}_{E_I}$ the basis of $\tau_{E_I}$. We define the possible worlds as $\Sigma=\mathcal{B}_{E_I}$, and the accessible relations $R_i$ as induced by $\tau_i=\tau_{E_I}$. The elements of $\mathcal{B}_{E_I}$ are global and atomic objects, such as global events. Our next identification is as follows:
\begin{itemize}
    \item Global events correspond to the basis topology of the mutual knowledge operator.
\end{itemize}
Therefore, any global description of an empirical model is given by the possible worlds $\Sigma=\mathcal{B}_{E_I}$ induced by the mutual knowledge. Therefore, according to Fine-Abramsky-Brandenburger Theorem \ref{FAB}, we can conclude the following.

\begin{result}
    Mutual knowledge is the knowledge that logically explains non-disturbing outcome-deterministic noncontextual empirical models.
\end{result}

\subsection{Local events follows from the topology of distributed knowledge}

In an analogous way to the previous argument, we can identify local sections as the elements of the basis of the topology induced by the mutual knowledge of their respective context. Since in a context $G$, every subcontext trusts each other, we have $D_G=E_G$: all distributed knowledge is described by mutual knowledge between the agents, with each of them having the information of all $G$. This is an example of how the trust relation influences the definition of distributed knowledge. Local events are the most refined propositions that can be made in the empirical model while respecting the contexts and the non-disturbing condition, and thus generate the most refined topology. On the other hand, the topology generated by the distributed knowledge operator is the most refined topology possible among a set of agents. Therefore, we can identify $\mathcal{B}_{D_I}$ as given by the local events. Our final identification of an element of a measurement scenario is as follows:
\begin{itemize}
    \item Local events correspond to the basis topology of the distributed knowledge operator.
\end{itemize}

The mapping of a measurement scenario to its respective multi-agent scenario is summarized in the following dictionary:

\begin{result}
    Given a measurement scenario $\left<X,\mathcal{U},(O_x)_{x\in X}\right>$ with the sheaf of events $\mathcal{E}$, the identification in the table below defines a multi-agent scenario with a set of agents $I$, trust relation $\stackrqarrow{E} {}$ and knowledge operators $D_I$ and $E_I$ induced by the basis topologies $\mathcal{B}_{D_I}$ and $\mathcal{B}_{E_I}$, respectively.
    \begin{table}[h!]
    \centering
    \begin{tabular}{| c | c |}
\hline\noalign{}
\textbf{Measurement scenario} & \textbf{Multi-agent scenario} \\
\hline\noalign{}
$X$ & $I$ \\ 
\hline\noalign{}
$\mathcal{U}$ & \makecell{$G\subset I$ with $\stackrqarrow{E} {}$\\ an equivalence relation}\\
\hline\noalign{}
$\mathcal{E}(X)$ & $\mathcal{B}_{E_I}$ \\
\hline\noalign{}
${\mathcal{E}(U)}_{U\in\mathcal{U}}$ & $\mathcal{B}_{D_I}$ \\
\hline\noalign{}
    \end{tabular}
\end{table}
\label{ResultadoTabelado}
\end{result}

\subsection{Logic contextuality is the failure of soundness}

How does contextuality manifest when viewed from a modal perspective and how does it relate to multi-agent paradoxes? The map presented in Result \ref{ResultadoTabelado} allows us to answer this question, as we will do below. For simplicity, we will limit ourselves to Boolean valuation functions, while making it clear that the probabilistic case with outcome-determinism follows from probabilistic distribution over the logical events identified by such Boolean valuation functions.

Previously, we saw that if an agent is terminal in the trust relationship, it has access to all the information of the other agents and sets of agents. Therefore, it can reconstruct the global view of the multi-agent scenario, and every other terminal agent will also agree with this description. What happens if the agents cannot agree on their global description? Well, one can argue that trust between agents and the sharing of information are not enough to access all the information of a scenario; $D_I \neq E_I$. In other words, the fundamental truth cannot be accessed by any agent individually. This is the key to characterizing contextual behavior, as we will see.

In an empirical model with Boolean valuation, the equation that represent noncontextuality is as follows
\begin{equation}
    \mu_{R}^{O^{U}}(A)=\sum_{\Lambda}p\left(\lambda\right)\prod_{x\in U}\mu_{R}^{O^{x}}(\rho'(U,x)(A)).
\end{equation}
This equation has every function as a Boolean function, thus outcome-determinism is satisfied. It also evaluates a formula $\phi$ by asking if, given all the possible worlds, one can semantically evaluate $\phi$ from them. Translating it in logic language, we get the following result:

\begin{result}[Logical contextuality]
    The definition of logical noncontextuality condition in logical notation takes the form of
\begin{equation}
\bigvee_{\lambda\in\Lambda}(\lambda\wedge(\lambda\to\phi))\vDash\phi,
\end{equation}
where due to the $\mathbf{S4}$ system always holds that
\begin{equation}
\bigvee_{\lambda\in\Lambda}(\lambda\wedge(\lambda\to\phi))\vdash\phi.
\end{equation}
Therefore, it is the violation of soundness which gives contextuality in the logical form.
\label{Logical noncontextuality}
\end{result} 

In Ref. \cite{Nurgalieva_2019}, it is stated that we have an ``inadequacy of modal logic in quantum settings'' precisely because we have multi-agent scenarios that exhibit logical paradoxes resulting from violations of soundness. Without the violation of soundness, multi-agent paradoxes do not arise\footnote{Since we trivially have completeness of the valuation in a multi-agent scenario, it is inconsistencies that give rise to multi-agent paradoxes. Soundness implies consistency, thereby avoiding the paradoxes.}. However, the identification we made in Result \ref{ResultadoTabelado} allows us to have an insight into what is actually happening, but which is left implicit in the literature. We can rewrite Result \ref{Logical noncontextuality} as follows:

\begin{result}[Logical contextuality with knowledge operators]
The definition of logical noncontextuality condition in logical notation and with explicit knowledge operators as identified in Result \ref{ResultadoTabelado} takes the form of
\begin{equation}
\bigvee_{E_I\lambda\in\mathcal{B}_{E_I}}(E_I\lambda\wedge(E_I\lambda\to K_i\phi))\vDash K_i\phi.
\end{equation}
This semantic equation does not always hold even if
\begin{equation}
\bigvee_{E_I\lambda\in\mathcal{B}_{E_I}}(E_I\lambda\wedge(E_I\lambda\to K_i\phi))\vdash K_i\phi
\end{equation}
holds syntactically. In other words, the logical form of contextuality follows from the violation of soundness when we define possible worlds as the elements of $\mathcal{B}_{E_I}$.
\label{Logical noncontextuality with knowledge operators}
\end{result}

The last equation of Result \ref{Logical noncontextuality with knowledge operators} states that if one can describe $\phi$ using elements from $\mathcal{B}_{E_I}$ that are true, then the agents know it. This differs from the semantic equation, where all $\phi$ must be described by it. This time, contextuality is not the failure of soundness that renders modal logic inadequate to deal with paradoxical behavior. Instead, it is the choice of the set of possible worlds as $\Sigma=\mathcal{B}_{E_I}$ that forces the logical violation to manifest in this way. There are insufficient possible worlds to adequately describe non-classical settings.

\subsection{Modal logic is suitable for
non-classical settings}

Paradoxes do not imply that modal logic is inadequate, as it is sound and complete in topological semantics. The problem is that we assume all descriptions must be consistent and global, which is not true. In other words, the worlds we are constructing in our scenario are too simplistic; they are defined by $E_I$, thereby disregarding any information beyond mutual knowledge.

We need to encode the information from all the empirical models, describing every detail of the agents, including their trust relationships, into the possible worlds. This is the case where the model exhibits lambda-dependence, where the worlds depend on the contexts. The elements of the basis must be the set of local events with their respective contexts, which, according to Result \ref{ResultadoTabelado}, is exactly given by $\tau_{D_I}$. With this new set of possible worlds, contextuality ceases to be the failure of soundness and becomes a matter of the empirical model not being described by classical worlds. 

\begin{result}[Logical contextuality as $E_I\neq D_I$]
Due to soundness and completeness of the topological semantics, 
\begin{equation}
\bigvee_{D_I\lambda\in\mathcal{B}_{D_I}}(D_I\lambda\wedge(D_I\lambda\to K_i\phi))\vDash K_i\phi.
\end{equation}
if and only if
\begin{equation}
\bigvee_{D_I\lambda\in\mathcal{B}_{D_I}}(D_I\lambda\wedge(D_I\lambda\to K_i\phi))\vdash K_i\phi.
\end{equation}
Therefore, contextuality is the difference between $E_I$ and $D_I$.
\label{Logical contextuality as}
\end{result}

The fundamental set of worlds is $\Sigma = \mathcal{B}_{D_I}$, while for a classical description, we are assuming that $\mathcal{B}_{E_I}$ is the set of words, a coarse-graining of the fundamental truth. In fact, Result \ref{Logical contextuality as} allows us to reach the following:

\begin{corollary}
    Modal logic is able to deal with the apparent violations if we do not restrict the knowledge to a mutual one, which we usually implicitly do.
\end{corollary}

\subsection{Limitations of the map for multi-agent scenarios}

The map constructed here has limitations in handling multi-agent scenarios with the contextual toolkit, and the cause of these limitations lies in the stringent constraints of the Sheaf Approach. To meet the criteria for analysis using the Sheaf Approach to contextuality, certain conditions must be fulfilled. 

Firstly, the agents should have only one measurement each, and these measurements must satisfy outcome-determinism, i.e., in quantum theory, they must be projection-valued measures\footnote{One can generalize the Sheaf Approach to deal with outcome-indeterminism \cite{Wester_2018}, but that is outside the scope of this article since the examples satisfy outcome-determinism.}. The trust relation defined between agents, and more generally between elements of the power set of the set of agents, must conform to the structure of contexts, specifically the equivalence trust relation. In particular, it must be symmetric, which prohibits the use of the map to address non-classicality in causal structures, a significant kind of generalization of the Wigner's friend scenario. Once these conditions are met, the measurement scenario becomes well-defined. 

To establish an empirical model, the events must satisfy the sheaf conditions, while the valuation must satisfy the no-disturbance condition. When these conditions are met, equivalence becomes possible, allowing one to explore the multi-agent paradox as contextuality using the tools of the Sheaf Approach.

\section{Three examples of contextuality in multi-agent scenarios}
\label{Multi-agent scenarios}

We will apply the previous results to analyze three well-known examples of multi-agent scenarios: Wigner's friend scenario, Frauchiger-Renner scenario, and Vilasini-Nurgalieva-del Rio scenario. Before delving into the actual examples, let's discuss common properties of these three scenarios.

The scenarios are formed by a set $I$ of agents. We will use names for Wigner, his friend (Alice), and their duplicated versions (Ursula and Bob respectively). 

The trust relation in all examples occurs between individual agents, simplifying the trust relations of Proposition \ref{Trust sets} to the usual definition in Axiom \ref{Trust Original}. Furthermore, the trust relation is symmetric in all examples. Therefore, we can represent the multi-agent scenarios as empirical models. They have contexts with two measurements defining a covering of contexts as a graph, where the measurements are identified as the vertices and the maximal contexts as the edges. All measurements have two outcomes, defining four local events.

All examples begin with a system in a certain initial state (Wigner's friend and Frauchiger-Renner scenarios a quantum states, Vilasini-Nurgalieva-del Rio scenario a Popescu-Rohrlich box). As shown in Ref. \cite{Vilasini_2019}, the act of an agent measuring the state defines an isomorphism between the system and the agent, allowing us to ignore the system and deal only with agents. In the same reference, it is also shown that the scenarios satisfy the property called "information-preserving memory update," which implies the same data being accessed by all the agents, i.e., there is trustworthiness among the agents who trust each other.

\subsection{Wigner's Friend scenario is noncontextual}

The standard Wigner's Friend scenario is defined with Alice $A$ performing a measurement on the system $R$, and with Wigner $W$ describing $R$ and $A$ in an entangled state due to her previous measurement. It asks for the different points of view between Alice and Wigner in the fundamental description of the nature of the probabilities involved. 

The scenario deals with an initial state $\ket{\phi}=\alpha\ket{0}+\beta\ket{1}$, with Alice's measurement in the basis $\{\ket{0},\ket{1}\}$. The problem here is where to put the Heisenberg's cut, before or after Alice. From Alice's point of view, after her measurement, the state is in a classical probability distribution $p_R(0)=\alpha^2$ and $p_R(1)=\beta^2$, and if she has already observed the result, it is certain to be one given eigenvalue. However, from Wigner's point of view, $R$ and $A$ defines a system $R\otimes A$ in a superposition being described by $\ket{\phi}$, thus the system and therefore Alice are described by a quantum superposition of states. 

There is no empirical contradiction here, as the classical probability distribution and the quantum state will give the same probabilities, and no discordance appears between Alice and Wigner. The problem that the Wigner's Friend scenario brings up is of an ontological nature: what is really happening with Alice?

Let's construct the empirical model of this scenario. Let Wigner perform a measurement in the system given by $R\otimes A$. We identify Alice and Wigner as the agents and we ignore the system $R$. There are two possibilities. The first one deals with Wigner's measurement being compatible with Alice's, which results in both of them trusting each other and defining a context
\begin{equation}
    A\leftrightsquigarrow W.
\end{equation}
The second possibility changes the basis in which Wigner performs his measurement to an incompatible one, for example $\ket{+}=\sqrt{\frac{1}{2}}\left(\ket{0}+\ket{1}\right)$ and $\ket{-}=\sqrt{\frac{1}{2}}\left(\ket{0}-\ket{1}\right)$. To Wigner, Alice's measurement is represented as a unitary transformation on $R\otimes A$ that changes Alice's state to a superposition. To him, the probability will be $p_{R\otimes A}(+)=\frac{\left(\alpha+\beta\right)^2}{2}$ and $p_{R\otimes A}(-)=\frac{\left(\alpha-\beta\right)^2}{2}$. To Alice, there is no probability at all if she already saw the measurement result and Wigner's measurement will just project the reduced state to his new basis. The problem here is that she knows her result, and Wigner erased it with his measurement, allowing no contradiction once the measurement erased Alice's memory as well\footnote{There is the problem of how to do it with a macroscopic entity, but this is not the point here. Our objective is not to address the possibility of implementing the scenario but rather to identify the source of paradoxical behavior in a generic and formal manner.}. 

Both possibilities allow analysis by considering only the measurement scenario. The first one has only one context, thus it must be noncontextual. The second one differs from the first by displaying two nonconnected contexts, making it a noncontextual empirical model. We can conclude:

\begin{result}
Wigner's friend scenario is represented by an empirical model with nonconnected contexts, therefore it is noncontextual and, consequently, shows no multi-agent paradox.
\end{result}

A realization of this result can be found in Ref. \cite{Lostaglio_2021}, where the authors construct a noncontextual model for Wigner's friend scenario.

\subsection{Frauchiger-Renner scenario is logic contextual}

The Frauchiger-Renner scenario \cite{Frauchiger_2018} starts with an entangled state 
\begin{equation}
    \ket{\phi}=\sqrt{\frac{1}{3}}\ket{0}\otimes\ket{0}+\sqrt{\frac{2}{3}}\ket{1}\otimes\sqrt{\frac{1}{2}}\left(\ket{0}+\ket{1}\right).
    \label{state}
\end{equation}
between two systems, $R$ and $S$, measured in the basis $\{\ket{0},\ket{1}\}$ by a respective friend, Alice $A$ and Bob $B$. After measurement, $R\otimes A$ and $S\otimes B$ are isomorphic to agents $A$ and $B$, respectively. As mentioned before, we can thus ignore the systems $R$ and $S$. The system $A$ measured by Ursula $U$ and the system $B$ measured by Wigner $W$ are measured in the basis $\{\ket{+},\ket{-}\}$, with $\ket{+}=\sqrt{\frac{1}{2}}\left(\ket{0}+\ket{1}\right)$ and $\ket{-}=\sqrt{\frac{1}{2}}\left(\ket{0}-\ket{1}\right)$. A locality argument can be used to describe who trusts whom. As we can ignore $R$ and $S$, the agents are Alice, Bob, Ursula, and Wigner. Trust is symmetric, and Alice's (Bob's) measurement is incompatible with Ursula's (Wigner's) measurement. Thus we get $A\leftrightsquigarrow W$, $U\leftrightsquigarrow B$, $A\leftrightsquigarrow B$, and $U\leftrightsquigarrow W$.

Once we are given the outcomes of the measurements, we can define the possible worlds using the knowledge operators of each agent. The topology induced by the mutual knowledge $E_I$ is generated by the elements of the basis, which consist of all $2^{4}$ combinations of the outcomes from the four agents. The outcome of a single agent is represented by the union of all the elements of this basis that contains such an outcome. The valuation is given by the initial state, but can only be calculated for the set of agents which mutually trust. $\ket{\phi}_{A\leftrightsquigarrow B}$ can be written exactly like equation \ref{state}, while the state that will be measured by $U\leftrightsquigarrow W$ will be
\begin{equation}
\begin{split}
    \ket{\phi}_{U\leftrightsquigarrow W}=&\sqrt{\frac{1}{12}}\left(\ket{+}+\ket{-}\right)\otimes\left(\ket{+}+\ket{-}\right)\\
    &+\sqrt{\frac{1}{3}}\left(\ket{+}-\ket{-}\right)\otimes\ket{+},
\end{split}
\end{equation}
and for $U\leftrightsquigarrow B$
\begin{equation}
\begin{split}
    \ket{\phi}_{U\leftrightsquigarrow B}=&\sqrt{\frac{1}{6}}\left(\ket{+}+\ket{-}\right)\otimes\ket{0}\\ &+\sqrt{\frac{1}{6}}\left(\ket{+}-\ket{-}\right)\otimes\left(\ket{0}+\ket{1}\right),
\end{split}
\end{equation}
and finally for $A\leftrightsquigarrow W$
\begin{equation}
    \ket{\phi}_{A\leftrightsquigarrow W}=\sqrt{\frac{1}{6}}\ket{0}\otimes\left(\ket{+}+\ket{-}\right)+\sqrt{\frac{2}{3}}\ket{1}\otimes\ket{+}.
\end{equation}
Calling the outcomes $+$ and $-$ of Ursula and Wigner respectively as $0$ and $1$, we can construct the table of the probabilities as shown in Table \ref{tab}.

\begin{table}[ht]
\begin{tabular}{|c|c|c|c|c|}
\hline
& \large{ $00$ } & \large{ $01$ } & \large{ $10$ } & \large{ $11$ } \\
\hline
\large{$A\leftrightsquigarrow B$} & \large{$\frac{1}{3}$} & \large{$0$} & \large{$\frac{1}{3}$} & \large{$\frac{1}{3}$} \\
\hline
\large{$A\leftrightsquigarrow W$} & \large{$\frac{1}{6}$} & \large{$\frac{1}{6}$} & \large{$\frac{2}{3}$} & \large{$0$} \\
\hline
\large{$U\leftrightsquigarrow W$} & \large{$\frac{3}{4}$} & \large{$\frac{1}{12}$} & \large{$\frac{1}{12}$} & \large{\color{red}$\frac{1}{12}$} \\
\hline
\large{$U\leftrightsquigarrow B$} & \large{$\frac{2}{3}$} & \large{$\frac{1}{6}$} & \large{$0$} & \large{$\frac{1}{6}$} \\
\hline
\end{tabular}
{\caption{Probabilities of the Frauchiger-Renner scenario. We call the outcomes $+$ and $-$ of Ursula and Wigner respectively as $0$ and $1$. Section $11$ of the context $U\leftrightsquigarrow W$ does not have a global event. Like Hardy's model, it shows possibilistic contextuality.}\label{tab}}
\end{table}

Let's follow the sequence of trust presented in Ref. \cite{Nurgalieva_2019}:
\begin{equation}
A\rightsquigarrow B\rightsquigarrow U\rightsquigarrow W\rightsquigarrow A.
\end{equation}
If Ursula measures $\ket{-}$, then Bob must measure $\ket{1}$ since $p(10|U\leftrightsquigarrow B)=0$. Consequently, Alice must measure $\ket{1}$ since $p(01|A\leftrightsquigarrow B)=0$, and Wigner must measure $\ket{+}$ since $p(11|A\leftrightsquigarrow W)=0$. However, as shown in Table \ref{tab}, Wigner can measure $\ket{-}$ since $p(11|U\leftrightsquigarrow W)=\frac{1}{12}$, contradicting Ursula's conclusion of $p(11|U\leftrightsquigarrow W)=0$. This is the violation presented in Ref. \cite{Frauchiger_2018}. The assumptions in Ref. \cite{Frauchiger_2018} are as follows:
\begin{itemize}
\item (Q) All agents use quantum theory.
\item (C) Agents can use the results from another agent.
\item (S) A measurement by an agent has an output defined for that agent.
\end{itemize}
The assumptions are an informal description of the definition of a multi-agent scenario, with assumption (S) defining the local events of each agent, assumption (C) connecting these events in global events, and assumption (Q) saying the valuation will be calculated by quantum mechanics. Such use of colloquial language was avoided by the formal construction of multi-modal logic multi-agent scenarios, but the conclusion remains: something in the assumptions must be weakened to explain the paradoxical behavior.

The empirical model can be constructed directly from Table \ref{tab}, which is equivalent to the previous equivalence. The possible worlds are defined as the basis of the topology generated by the mutual knowledge $E_I$ and identified as the global events. The empirical model that results from the valuation is non-disturbing, as one can directly verify, and contextual\footnote{Using the noncontextual fraction \cite{Abramsky_2017}, one can find $NCF=\frac{5}{12}$.}. The possibilistic bundle diagram of Table \ref{tab} is given by Figure \ref{Bell}. The section $11$ of the context $U\leftrightsquigarrow W$ does not have a possibilistic global event, and by imposing Ursula's conclusion, $p(11|U\leftrightsquigarrow W)=0$, the induced possibilistic empirical model becomes noncontextual, showing that it is the cause of the possibilistic contextuality and equivalently the cause of the multi-agent paradox in the Frauchiger-Renner scenario. Thus we have the following conclusion:

\begin{result}
    Frauchiger-Renner scenario is mapped as an empirical model presenting logic contextuality, the result of its multi-agent paradox with quantum origin.
\end{result}

As mentioned earlier, from the Result \ref{Logical contextuality as}, we conclude that if we extend the possible worlds to encompass all local sections indexed by their contexts, the paradox also disappears. However, indexing leads to lambda-dependence, a non-classical property that ultimately embodies contextuality. This affirms a claim made in Ref. \cite{Nurgalieva_2019} stating that the inclusion of contexts as data in propositions avoids logical contradictions in the Frauchiger-Renner scenario. 

\begin{figure}[ht]
    \centering
    \scalebox{0.3}{\begin{tikzpicture}
\draw [ultra thick] (-9,12) -- (-3,10) node[right] {\Huge{$\ B$}};
\draw [ultra thick] (-3,10) -- (0,13) node[right] {\Huge{$\ U$}}; 
\draw [ultra thick] (0,13) -- (-6,15)
node[left] {\Huge{$W\ $}}; 
\draw [ultra thick] (-6,15) -- (-9,12) node[left] {\Huge{$A\ $}};
\filldraw [black] (-9,12) circle (4pt);
\filldraw [black] (-3,10) circle (4pt);
\filldraw [black] (0,13) circle (4pt);
\filldraw [black] (-6,15) circle (4pt);
\draw[loosely dotted, ultra thick] (-9,12) -- (-9,21);
\draw[loosely dotted, ultra thick] (-3,10) -- (-3,19);
\draw[loosely dotted, ultra thick] (0,13) -- (0,22);
\draw[loosely dotted, ultra thick] (-6,15) -- (-6,24);
\filldraw [black] (-9,21) circle (4pt);
\filldraw [black] (-3,19) circle (4pt);
\filldraw [black] (0,22) circle (4pt);
\filldraw [black] (-6,24) circle (4pt);
\filldraw [black] (-9,19) circle (4pt);
\filldraw [black] (-3,17) circle (4pt);
\filldraw [black] (0,20) circle (4pt);
\filldraw [black] (-6,22) circle (4pt);
\draw [ultra thick] (-3,17) -- (0,20) node[right] {\Huge{$\ 0$}}; 
\draw [ultra thick] (0,20) -- (-6,22); 
\draw [ultra thick] (-6,22) -- (-9,21);
\draw [ultra thick] (-9,21) -- (-3,17);
\draw [ultra thick] (-3,19) -- (0,22) node[right] {\Huge{$\ 1$}}; 
\draw [ultra thick] (-3,19) -- (0,20);
\draw [red, ultra thick] (0,22) -- (-6,24); 
\draw [ultra thick] (-6,24) -- (-9,19);
\draw [ultra thick] (-9,21) -- (-3,19);
\draw [ultra thick] (-9,19) -- (-3,17);
\draw [ultra thick] (-6,22) -- (-9,19);
\draw [ultra thick] (0,22) -- (-6,22); 
\draw [ultra thick] (0,20) -- (-6,24); 
\end{tikzpicture}}
    \caption{Possibilistic bundle of the Frauchiger-Renner scenario.}
    \label{Bell}
\end{figure}
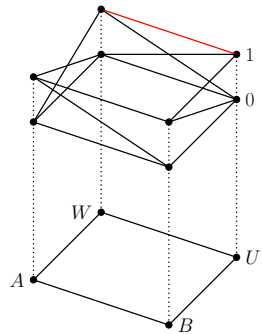

\subsection{Vilasini-Nurgalieva-del Rio scenario is strongly contextual}

Another example is the Vilasini-Nurgalieva-del Rio scenario \cite{Vilasini_2019}. It generalizes the conditions for multi-agent paradoxes in generalized probability theories using modal logic and explicitly constructs a paradox for the box world. The construction of the multi-agent scenario in Ref. \cite{Vilasini_2019} has equivalent assumptions to those described at the beginning of the section, in addition to requiring treatment of states, effects, and the channels that the trust relation defines between the agents, all using the formalism of Generalized Probabilistic Theories \cite{Janotta2014, Muller2021, Selby2021} applied in the box world defined with Popescu-Rohrlich boxes \cite{Popescu1994}. Here, we will limit ourselves to what is necessary for the calculation of valuations and the construction of the empirical model, leaving the original article as a reference for the detailed construction of the multi-agent scenario.

The structure of the agents, trust relation, and the possible worlds is identical to the one presented in the Frauchiger-Renner scenario. We have two systems $R$ and $S$, two friends Alice $A$ and Bob $B$, Wigner $W$ and Ursula $U$, with a symmetric trust relation given by $A\leftrightsquigarrow W$, $U\leftrightsquigarrow B$, $A\leftrightsquigarrow B$, and $U\leftrightsquigarrow W$, while the systems can be ignored due to their isomorphism with the friends. Each measurement will have two outcomes, defining the $2^4$ possible worlds given by $E_I$. 

The valuation follows the initial state given by the sharing of a Popescu-Rohrlich box between $R$ and $S$, thus satisfying $X_{i}X_{j}=x_{i}\oplus_{mod 2} x_{j}$ with $X_{i}$ measurements and $x_{i}$ outcomes. The authors of \cite{Vilasini_2019} show that all pairs of agents trusting each other can be understood as being correlated by Popescu-Rohrlich boxes. By using trustworthy and fixing the conditions $X_{U}=X_{A}\oplus_{mod2}1$, $X_{W}=X_{B}\oplus_{mod2}1$, the measurements $X_{A}=X_{B}=0$ and the outcomes $x_{i}\in\{0,1\}$, we can propagate the initial correlation between $R$ and $S$ to obtain the possibilistic values presented in Table \ref{tab2}. 

\begin{table}[ht]
  \begin{tabular}{|c|c|c|c|c|}
\hline
& \large{ $00$ } & \large{ $01$ } & \large{ $10$ } & \large{ $11$ } \\
\hline
\large{$A\leftrightsquigarrow B$} & \large{$1$} & \large{$0$} & \large{$0$} & \large{$1$} \\
\hline
\large{$A\leftrightsquigarrow W$} & \large{$1$} & \large{$0$} & \large{$0$} & \large{$1$} \\
\hline
\large{$U\leftrightsquigarrow W$} & \large{$0$} & \large{\color{red}$1$} & \large{\color{red}$1$} & \large{$0$} \\
\hline
\large{$U\leftrightsquigarrow B$} & \large{$1$} & \large{$0$} & \large{$0$} & \large{$1$} \\
\hline
\end{tabular}
{\caption{Prossibilities of the Vilasini-Nurgalieva-del Rio scenario. It defines the well-known Popescu-Rohrlich box empirical model, showing the Liar Cycle paradox with four agents. It is strong contextual once all local sections show violations, thus making it stronger than the previous example.}\label{tab2}}
\end{table}

As shown in Ref. \cite{Vilasini_2019}, all agents find a contradiction in any chosen sequence of agents, presenting a stronger violation than the one presented by the Frauchiger-Renner scenario. Using the same assumptions as the Frauchiger-Renner scenario, with the necessary modification that in (Q) the agents use the box world, at least one of them would need to be violated to explain the paradoxical behavior.

The identification with an empirical model follows the exact same construction to the one for the Frauchiger-Renner scenario, but now we are dealing with possibilistic values, thus allowing a faithful representation of Table \ref{tab2} as the bundle diagram in Figure \ref{PR}. It defines the well-known Popescu-Rohrlich box empirical model, showing the Liar Cycle paradox with four agents. Once all local sections show violations,thus making it stronger than the empirical model of the previous example, we have the following conclusion:

\begin{result}
    Vilasini-Nurgalieva-del Rio scenario is mapped as an empirical model known as Popescu-Rohrlich box empirical model, a main example of strong contextuality, the result of its multi-agent paradox with post-quantum origin.
\end{result}

Since they share the same measurement scenario, both the Frauchiger-Renner scenario and the Vilasini-Nurgalieva-del Rio scenario have the same set of possible worlds given by $\mathcal{B}_{D_I}$. Similarly to the previous example, we can use Result \ref{Logical contextuality as} to rectify the paradoxical behavior at the expense of lambda-dependence.

\begin{figure}[ht]
    \centering
    \scalebox{0.3}{\begin{tikzpicture}
\draw [ultra thick] (-9,12) -- (-3,10) node[right] {\Huge{$\ B$}};
\draw [ultra thick] (-3,10) -- (0,13) node[right] {\Huge{$\ U$}}; 
\draw [ultra thick] (0,13) -- (-6,15)
node[left] {\Huge{$W\ $}}; 
\draw [ultra thick] (-6,15) -- (-9,12) node[left] {\Huge{$A\ $}};
\filldraw [black] (-9,12) circle (4pt);
\filldraw [black] (-3,10) circle (4pt);
\filldraw [black] (0,13) circle (4pt);
\filldraw [black] (-6,15) circle (4pt);
\draw[loosely dotted, ultra thick] (-9,12) -- (-9,21);
\draw[loosely dotted, ultra thick] (-3,10) -- (-3,19);
\draw[loosely dotted, ultra thick] (0,13) -- (0,22);
\draw[loosely dotted, ultra thick] (-6,15) -- (-6,24);
\filldraw [black] (-9,21) circle (4pt);
\filldraw [black] (-3,19) circle (4pt);
\filldraw [black] (0,22) circle (4pt);
\filldraw [black] (-6,24) circle (4pt);
\filldraw [black] (-9,19) circle (4pt);
\filldraw [black] (-3,17) circle (4pt);
\filldraw [black] (0,20) circle (4pt);
\filldraw [black] (-6,22) circle (4pt);
\draw [ultra thick] (-3,17) -- (0,20) node[right] {\Huge{$\ 0$}}; 
\draw [ultra thick] (-3,19) -- (0,22) node[right] {\Huge{$\ 1$}};
\draw [ultra thick] (-9,21) -- (-3,19);
\draw [ultra thick] (-9,19) -- (-3,17);
\draw [ultra thick] (-6,22) -- (-9,19);
\draw [ultra thick] (-6,24) -- (-9,21);
\draw [red, ultra thick] (0,22) -- (-6,22); 
\draw [red, ultra thick] (0,20) -- (-6,24); 
\end{tikzpicture}}
    \caption{Possibilistic bundle of the Vilasini-Nurgalieva-del Rio scenario.}
    \label{PR}
\end{figure}
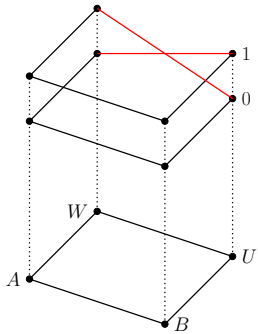

\section{Discussion}
\label{Discussion}

Multi-agent paradoxes and the phenomenon of contextuality are examples of systems with local, marginal access to a given global whole that cannot be adequately explained in a classical manner. Alfred Korzybski's statement, ``a map is not the territory'' \cite{korzybski1933science}, serves as a reminder that even when dealing with the classical world, we need to keep in mind that we do not have the capacity to discern all the details of the ``territory''; we are only capable of constructing ``maps.'' An agent's knowledge only extracts parts of a ``map'' of the ``territory'' we call physical reality. This aligns with Refs. \cite{Del_Santo_2019, Gisin_2019, Gisin_2021} regarding the ontological assumptions made in physical theories when using mathematical objects that do not align with a purely empiricist viewpoint and how such assumptions can be understood as the hidden variables of classical theory. The non-classical world imposes an even greater limitation on us by isolating us in classical islands, which, much like the concept of charts in the theory of manifolds, can only be connected through an atlas but lose something essential that a map of a single chart possesses. In a non-classical world, we not only have to remember that a map is not the territory, but also that the map is merely a chart in an atlas. There is something beyond what we can individually perceive with our classically limited perspective, making Korzybski's statement even more imperative.

The main point we obtained was that multi-modal logic can be used to handle non-classical scenarios, provided that due care is taken with the agents' knowledge. During our journey to build the map between empirical models and multi-agent scenarios, we encountered other results. Here are the results we have obtained in this paper:
\begin{itemize}
    \item We generalized the concept of trust to also apply to sets of agents.
    \item We have identified a construction of a fundamental truth from the trust relation between agents, which is used in the literature to weaken the Truth Axiom, thus recovering the Truth Axiom from trust.
    \item We have translated the components of an empirical model (measurements, contexts, events) into the components of a multi-agent scenario (measurement by an agent, equivalence trust relation, elements of the basis topology of a knowledge operator), exposing the limitation of such a mapping to describe a generic multi-agent scenario as an empirical model.
    \item We have shown that contextuality is the violation of soundness, precisely the violation that causes multi-agent paradoxes.
    \item We have demonstrated that contextuality only appears because of the imposition of the mutual knowledge operator as the generator of events, imposing exactly the classicality captured by noncontextuality.
    \item By allowing the operator that generates events to be that of distributed knowledge, we have shown that we recover soundness at the cost of lambda-dependence, proving a generalized version of the conjecture in Ref. \cite{Nurgalieva_2019} that says the inclusion of contexts as data of propositions avoids logical contradictions.
    \item We have translated the three main examples of multi-agent scenarios, Wigner's Friend, Frauchiger-Renner, and Vilasini-Nurgalieva-del Rio scenarios, into empirical models and identified their types of contextuality that generate their multi-agent paradoxes.
\end{itemize}

These results enable the creation of new multi-agent paradoxes and the application of mathematical tools from the Sheaf Approach to contextuality in these scenarios. The discussed examples in section \ref{Multi-agent scenarios} illustrate how the map can identify paradoxes in a clear manner, with the possibility of even quantifying them using tools already known in the literature on contextuality. It's interesting to highlight the types of contextuality that emerge from the explored example scenarios. In the Wigner's Friend scenario, no contextuality appears, clearly indicating its non-empirical nature. On the other hand, the Frauchiger-Renner scenario exhibits logical contextuality, demonstrating its already noted similarity to the well-known Hardy's paradox. The Vilasini-Nurgalieva-del Rio scenario demonstrates strong contextuality, as expected since it was constructed with Popescu-Rohrlich boxes. The results also facilitate the translation of contextual empirical models into multi-agent scenarios with probabilistic paradoxes, and possibilistic contextual empirical models into logical multi-agent paradoxes.

The examples show us that paradoxes and contextuality are the same phenomenon, at least in scenarios that accept the inverse map, and through their valuations, it is evident that there is more knowledge than the mutual one. They point to the argument we use to construct events, which is that we cannot define the worlds our logic will work out, but by the knowledge we can explore and refine the worlds we can have access to in a more empiricist and relational sense. Distributed knowledge is the finest way to understand what is happening, as it encodes all the data in the propositions, safeguarding our classical logic from the non-classical nature of phenomena. It also demonstrates that we have more data than classical mutual knowledge, more possible worlds, and, as we can observe today with quantum technology, more resources to explore. 

An immediate path for future research would be to seek generalizations of the Sheaf Approach, as explored in Ref. \cite{Gogioso_2021, gogioso2023topology, abramsky2023combining}, in order to expand the scope of application of the map constructed here. The ultimate goal is to enable the representation of every multi-agent scenario with empirical models, allowing the use of sheaf theory tools to investigate their non-classicality. Another avenue is to leverage the relationship that the Sheaf Approach has with other approaches to the phenomenon of contextuality, as organized in Ref. \cite{masse2021problem} and references therein, to describe multi-agent paradoxes in these languages. In particular, utilizing their different domains of application to further extend the map constructed here. Both of these aforementioned paths would be highly valuable for formalizing the language of protocols that present extended versions of the Wigner's friend scenario, especially those with a causal structure \cite{schmid2023review}. While an extension of the map between empirical models and multi-agent scenarios is not yet constructed, it would be interesting to identify and develop examples that allow analysis through the standard Sheaf Approach, with ideal candidates already existing in Ref. \cite{Leegwater2022, schmid2023review, walleghem2023extended}. Furthermore, the use of multi-modal logic and the construction done here using topological semantics and pointless topology may shed light on the formalization of interpretations of quantum theory, such as the relational interpretation in Ref. \cite{sep-qm-relational, Lawrence2023relativefactsof}.

\begin{acknowledgments}
The author thanks the MathFoundQ – UNICAMP – Mathematical Foundations of Quantum Theory Group, in special Marcelo Terra Cunha, for the conversations in the preparation of this manuscript, and Vinicius Pretti Rossi for conversation about Wigner's Friend scenarios. Special thanks to Rafael Wagner, Alisson Tezzin, and Rafael Rabelo for their comments and discussions on the final version of the manuscript and its results.

This study was financed in part by the Coordenação de Aperfeiçoamento de Pessoal de Nível Superior - Brasil (CAPES) - Finance Code 001.
\end{acknowledgments}

\appendix

\section{Modal Logic}
\label{Modal Logic}

A modal logic is defined with a set $\Omega$ of propositional variables and the usual set of connectives $\neg$ (``not''), $\wedge$ (``and''), $\vee$ (``or''), $\leftrightarrow$ (``if and only if''), $\to$ (``if . . . then''), besides the use of parentheses.

In addition to the usual connectives, a modal logic has a modal operator called ``possibility'' $\Diamond$. When combined with $\neg$ one can define the modal operator ``necessity'' $\Box$ as $\neg\Diamond\neg$.\footnote{One can also start with $\Box$ and to define $\Diamond$ as $\neg\Box\neg$, but certain care must be taken when defining the modal operator in this dual manner\cite{sep-phil-multimodallogic}.} 

When dealing with a set of agents indexed by a finite set $I\ni i$, one can define $\Diamond_i$ (and consequently $\Box_i$) as the necessity modal operator from the point of view of agent $i$. This defines a multi-modal logic, with one modal logic for each agent, but all of them agreeing on the usual propositional logic structure.

Once the set of propositional variables $\Omega$ and symbols are defined, one can define the formulas as follows:
\begin{itemize}
\item All the propositional variables are formulas.
\item If $A$ is a formula, then $\neg A$, $\Diamond A$, and $\Box A$ are formulas.
\item If $A$ and $B$ are formulas, then $(A\wedge B)$, $(A\vee B)$, $(A\leftrightarrow B)$, and $(A\to B)$ are also formulas.
\item There are no other formulas.
\end{itemize}
The collection of propositions $\Phi$ is defined by the possible formulas.

\subsection{Kripke Semantics of multi-modal logic}

A Kripke frame $\left<\Sigma,R\right>$ is a pair consisting of a non-empty set of states or worlds $\Sigma$ and a binary relation $R$ on $\Sigma$, called the accessibility relation, such that $aRb$ means ``$b$ is possible given $a$'' or ``$b$ is accessible by $a$''.

A relational structure $\left<\Sigma,\{R_{i}\}_{i\in I}\right>$ is a finite set of Kripke frames with the same $\Sigma$, where each $R_{i}$ is given by an agent $i$. In other words, $aR_{i}b$ is understood as ``$b$ is possible given $a$ in the point of view of agent $i$'' or ``$b$ is accessible by $a$ in the point of view of agent $i$''.

A Kripke structure $M=\left<\Sigma,\{R_{i}\}_{i\in I}, \nu\right>$ is a relational structure $\left<\Sigma,\{R_{i}\}_{i\in I}\right>$ equipped with a Boolean valuation $\nu:\Omega\to\mathscr{P}(\Sigma)$, with $\mathcal{P}(\Sigma)$ the power set of $\Sigma$, that indicate the worlds where a proposition variable is true: given $A\in\Omega$, $\nu(A)\in\mathcal{P}(\Sigma)$ is the set of worlds where $A$ is true. The valuation of a generic proposition in $\Phi$ obeys the ordinary rules of propositional logic for each world, plus rules to the modal operators as we will see.


\subsection{Rules, soundness and completeness}

Given a Kripke structure $M=\left<\Sigma,\{R_{i}\}_{i\in I}, \nu\right>$ with possible worlds $w\in\Sigma$ we define $M,w\vDash\phi$ as the proposition $\phi$ being true for the world $w\in\Sigma$ of the Kripke structure $M$. Equivalently, we write $M,w\vDash\phi$ if $w\in\nu(\phi)$ in the Kripke structure $M$.

Sentences, also known as closed formulas, are formulas without free variables. Let $Q$ be a sentence. The symbol $Q\vDash \phi$, where $\phi\in\Phi$, can be read as ``$Q$ semantically entails $\phi$'', meaning that if $Q$ is true, then $\phi$ is also true. We can have a finite set of sentences $Q_1,...,Q_n$ semantically entailing $\phi$, $Q_1,...,Q_n\vDash\phi$, which reads as ``if the sentences $Q_1,...,Q_n$ are true, then $\phi$ is true.'' 

Another symbol $\vdash$ can be read as ``$Q$ syntactically entails $\phi$'', meaning $Q$ proves $\phi$. Again, we can have a finite set of sentences $Q_1,...,Q_n$ syntactically entailing $\phi$, in symbols $Q_1,...,Q_n\vdash\phi$, which is read as ``the sentences $Q_1,...,Q_n$ prove $\phi$.'' 

The ordinary rules of propositional logic hold here for each world, and additional rules for the modal operators in Kripke semantics are added:
\begin{itemize}
    \item $(M,w\vDash\Box\phi)\leftrightarrow\forall u(wRu\to(M,u\vDash\phi))$.
    \item $(M,w\vDash\Diamond\phi)\leftrightarrow\exists u(wRu\wedge(M,u\vDash\phi))$.
\end{itemize}
A system satisfies completeness (also called semantic completeness) if $Q\vDash\phi$ implies $Q\vdash\phi$, and a system satisfies soundness if $Q\vdash\phi$ implies $Q\vDash\phi$.


\subsection{Knowledge operators}
\label{Knowledge}

The valuation $\nu$ being unique for all agents reflects the philosophical statement that truth is independent of any agent; it is absolute. This can be understood as a strong axiom to determine the distinction between knowledge and belief, with the former being a direct consequence of truth and the latter not needing any relation to it. This definition of knowledge is Plato's ``justified true belief.'' However, as one can readily see, different agents have different knowledge, which is a coarse-graining of the fundamental truth. Therefore, for multi-agent scenarios, we must use the knowledge of each agent to valuate propositions.

One can define, for an agent, the basic modal operator of epistemic logic $K$ that means ``it is known that''. Let $R(w)=\{u|wRu\}$, and for $A\subseteq\Sigma$ denote $M,A\vDash\phi$ as $M,u\vDash\phi$ for all $u\in A$. Them, in Kripke semantics, one add a new rule to define knowledge:
\begin{itemize}
    \item $(M,w\vDash K\phi)\leftrightarrow(M,R(w)\vDash \phi)$.
\end{itemize}
In the case of multiple agents indexed by a set $I$, one can define an operator $K_i$ for each agent $i$, where $K_i\phi$ can be read as ``agent $i$ knows that $\phi$''. We need to add a new item to the list of formulas:
\begin{itemize}
    \item If $A$ is a formula, then $K_iA$ for all $i\in I$ is a formula.
\end{itemize}

To preserve the truth by the knowledge operators, one imposes the Knowledge generalization rule, also known as $\mathbf{N}$ and Necessitation Rule, that says for a Kripke structure $M$ and any $\phi\in\Phi$, we have
\begin{equation}
(M,w\vDash\phi) \forall w\to (M,w\vDash K_i\phi)\forall i.
\end{equation}
This rule can be written as well for modal operators,
\begin{equation}
(M,w\vDash\phi) \forall w\to (M,w\vDash \Box\phi).
\end{equation}

There are two more modal operators dealing with knowledge of a subset of agents $U\subset I$ that are interesting to us. Mutual or common knowledge $E_G$ means ``every agent in $G$ knows''. Formally, for all $\phi$, we define the mutual knowledge operator as follows:
\begin{equation}
E_G\phi=\bigwedge_{i\in U}K_{i}\phi,
\end{equation} 
which defines a relation
\begin{equation}
R_{E_{G}}=\bigcup_{i\in G}R_i 
\end{equation} 
that allows the addition of the following rule in the Kripke semantics:
\begin{itemize}
    \item $(M,w\vDash E_G\phi)\leftrightarrow(M,R_{E_G}(w)\vDash \phi)$.
\end{itemize}
Distributed knowledge $D_G$ means ``it is distributed knowledge to the whole $U$'', not just describing the knowledge of individual agents but all knowledge combined of $U$ as an entity itself. Formally, for all $\phi$, we define mutual knowledge operator as follows:
\begin{equation}
D_G\phi=\bigvee_{i\in U}K_{i}\phi,
\end{equation}
which defines a relation
\begin{equation}
R_{D_{G}}=\bigcap_{i\in G}R_i
\end{equation}
that allows the addition of the following rule in the Kripke semantics:
\begin{itemize}
    \item $(M,w\vDash D_G\phi)\leftrightarrow(M,R_{D_G}(w)\vDash \phi)$.
\end{itemize}

\subsection{Axioms of system \textbf{S4}}

Different axioms can be imposed on the accessibility relation of a frame (Frame Conditions) that equivalently\footnote{They follow from the preservation of such properties on the accessible worlds of each world.} result in properties of modal (Modal Axioms) and knowledge (Axioms of Knowledge) operators, thus defining different systems of modal logic\cite{sep-logic-modal,sep-logic-epistemic}.

\begin{axiom}[Distribution Axiom or $\mathbf{K}$]
It holds true for any frame. For modal operators, we have that for any $\psi,\phi\in\Phi$ it holds that
\begin{equation}
(\Box(\psi\to\phi))\to(\Box\psi\to \Box\phi)
\end{equation}
while for knowledge operators, for any $\psi,\phi\in\Phi$, we have 
\begin{equation}
(K_{i}\phi\wedge K_{i}(\phi\to\psi ))\to K_{i}\psi.
\end{equation}
\end{axiom}

System $\mathbf{K}$ is the simplest kind of logic described by Kripke semantics and establishes modus ponens for each world. An equivalent way to write it as a Modal Axiom is
\begin{equation}
\Box(\phi\wedge (\phi\to\psi ))\to \Box\psi,
\end{equation}
in a similar format to the respective Axiom of Knowledge. Normal Modal System is defined as a system $\mathbf{K}$ satisfying Rule $\mathbf{N}$.

\begin{axiom}[Truth Axiom, or $\mathbf{T}$, or $\mathbf{M}$]
For any frame and $\phi\in\Phi$:
\begin{itemize}
    \item (Frame Condition) The accessibility relation is reflexive.
    \item (Modal Axiom) $\Box\phi\to\phi.$
    \item (Axiom of Knowledge) $K_i\phi\to\phi.$
\end{itemize}
\end{axiom}

As a result of this axiom, one can show that $\phi\to\Diamond\phi$ holds. System $\mathbf{T}$ (also known as System $\mathbf{M}$) is defined as a System $\mathbf{K}$ satisfying the Truth Axiom.

\begin{axiom}[Positive Introspection Axiom or $\mathbf{4}$]
For any frame and $\phi\in\Phi$:
\begin{itemize}
    \item (Frame Condition) Accessibility relation is transitive.
    \item (Modal Axiom) $\Box\phi\to\Box\Box\phi.$
    \item (Axiom of Knowledge) $K_{i}\phi\to K_{i}K_{i}\phi.$
\end{itemize}
\end{axiom}

A result of this Axiom is that holds $\Diamond\Diamond\phi\to\Diamond\phi$. System $\mathbf{S4}$ is defined as a System $\mathbf{T}$ satisfying Axiom $\mathbf{4}$.\footnote{Another important axiom known as Negative Introspection Axiom or $\mathbf{5}$ is the imposition of symmetry of the accessibility relation, resulting for any $\phi\in\Phi$ the validity of $\neg K_{i}\phi\implies K_{i}\neg K_{i}\phi$ for knowledge operators and $\Diamond\phi\to\Box\Diamond\phi$ for modal operators. System $\mathbf{S5}$ is defined as a System $\mathbf{S4}$ satisfying Axiom $\mathbf{5}$, and is exactly the system where the accessibility relation is an equivalence relation. Usually one drops $\mathbf{5}$ once when an agent does not know something, it is hard to such agent judge its own lack of knowledge.}

\subsection{Topological semantics of multi-modal logic}

A natural semantics for the system $\mathbf{S4}$ is the topological \cite{awodey_kishida_2008,EConiglio2017-ECOMLS,https://doi.org/10.48550/arxiv.math/0703106,Awodey2012-AWOTCO-2}.

\begin{definition}
A topological model is a pair $(T,\nu)$ where $T=(X,\tau)$ is a topological space and a function $\nu:\Phi\to\mathscr{P}(X)$, called interpretation, that satisfies for any $\phi,\psi\in\Phi$
\begin{equation*}
    \begin{split}
        \nu(\phi\wedge\psi)&=\nu(\phi)\cap \nu(\psi)\\
        \nu(\phi\vee\psi)&=\nu(\phi)\cup \nu(\psi)\\
        \nu(\neg\phi)&=\nu(\phi)^{\complement}\\
        \nu(\Box\phi)&=\interior{\nu(\phi)}\\
        \nu(\Diamond\phi)&=\closure{\nu(\phi)},
    \end{split}
\end{equation*}
with $\interior{A}$, $\closure{A}$ and $A^{\complement}$ respectively the topological interior, closure and complement of $A\in\mathscr{P}(X)$.
\end{definition}

The elements of $X$ are the worlds, and $\nu$ can be understood as the valuation in the topological semantics, by giving the set of worlds where a formula is true, $M,w\vDash\phi$ if and only if $w\in\nu(\phi)$. Also, for any two formulas $\phi,\psi\in\Phi$ one can prove $\phi\vdash\psi$ if and only if $\nu(\phi)\subseteq\nu(\psi)$. One can show that this semantics imposes system $\mathbf{S4}$ to the logic. In this sense, the system $\mathbf{S4}$ is said to be the logic of topological spaces. 

An Alexandrov topological space is a topological space where every point of the space has a minimal neighborhood. Alexandrov topologies can also be defined as topological spaces where arbitrary intersections of open sets are open sets. In particular, any finite topology, i.e., only finitely many open sets, is an Alexandrov topology. A well-known result is the equivalence between Kripke and topological semantics for Alexandrov topological spaces \cite{10.2307/1969080}:

\begin{theorem}
For any Alexandrov topological space $(X,\tau)$ there exists a binary relation $R$ such that for any Boolean valuation $\nu$ and for any formula $\phi\in\Phi$, $(X,\tau,\nu),w\vDash\phi$ if and only $(X,R,\nu),w\vDash\phi$. For any transitive reflexive frame $(\Sigma,R)$, equivalently a frame satisfying $\mathbf{S4}$, there exists a topology $\tau$ on $\Sigma$ such that for any valuation $\nu$ and for any formula $\phi\in\Phi$, $(\Sigma,R,\nu),w\vDash\phi$ if and only $(\Sigma,\tau,\nu),w\vDash\phi$.
\end{theorem}

This result \cite{10.1007/978-3-540-72734-7_12} follows from the identification of the accessibility relation $R$ with the specialization pre-order $\leq$, $x\leq y$ if and only if $\forall U\in\tau$ we have $(x\in U)\to (y\in U)$, which turns $(X,\leq)$ into a poset. Such a relation defines a topology generated by the basis of open sets $U^x=\{y|y\leq x\}$, which is an equivalent definition of the Alexandrov topology\footnote{This is the upper Alexandrov topology, and one can think of it as defining open sets as generated by the causal past cones of points. The lower Alexandrov topology, with the basis $U_y=\{x|y\leq x\}$, is given by the future causal cones \cite{2009.12646}.}, and any Alexandrov topology has such natural pre-order that defines the semantics satisfying $\mathbf{S4}$. Also, the system $\mathbf{S4}$ satisfies completeness and soundness in relation to the topological semantics of Alexandrov topological spaces.

\section{Sheaf approach to contextuality}
\label{Sheaf approach}

Contextuality is, informally, the property of a physical system that cannot be explained classically, where this classicality is thought of as an ontological reality that is coarse-grained to the system\footnote{See Ref. \cite{budroni2022quantum,masse2021problem} for a general revision of contextuality.}. The Sheaf approach, which chooses as the fundamental objects the measurements and their joint measurability, codifies the system as a category. The set of outcomes of each measurement and the distribution on it are codified by functors. Contextuality will appear as a property of the functor over a given system. 

\subsection{Measurements}

As in Ref. \cite{Terra_2019_FibreBundle} we will define measurements in a generic manner by its results, also called events. For simplicity we will keep all the sets with finite elements.

\begin{definition}
    A measurement $M$ is a set of labels $\{s^i\}_{i=1}^n$ for the possible $n$ events. We will denote by $M[s^i]$ the event $s^i$ of the measurement $M$.
\end{definition}

The formal definition of a measurement in a physical theory depends on the theory in which it is constructed. A well-known example of measurement is given by a positive operator-valued measure in quantum theory. Its events are given by its outcomes, the effects that sum to the identity.

\subsection{Compatibility between measurements}

The measurements are organized as a covering through compatibility, or joint measurability. One can understand compatibility as classicality at the ontological level. It is a property between measurements and their respective outcomes given by the theory, not dealing with distributions over outcomes. It is a condition of classicality stronger than the concept of noncontextuality, which is epistemic in nature and arises from classicality at the level of distributions over outcomes, as we will see later with the sheaf approach. Compatibility imposes the existence of a ``mother'' measurement, such that our accessible measurements have origin by classical post-processing\footnote{Both notions are equivalent, see Ref. \cite{Filippov_2017}.}.

\begin{definition}
Let $\{M_{k}\}_{k=1}^m$ be a set of measurements with a respective set of events $O^{(k)}=\{s^{(k)}\}$, they are jointly measurable if there exists a measurement $G$ with set of events $O^{(1)}\times ...\times O^{(m)}$ satisfying
\begin{equation}
    M_{k}[s^{(k)}]=\sum_{s^{(j)}:j\neq k}G[s^{(1)},...,s^{(k)},...,s^{(m)}]
\end{equation}
for all $k$. 
\end{definition}

Therefore, a set of compatible measurements allows the existence of a measurement that can recover the original measurements when appropriately marginalized. As shown in Ref. \cite{Heinosaari_2010}, in quantum theory commuting implies jointly measurable, and the inverse holds if the measurements are sharp. In this paper all the quantum measurements are sharp.

\subsection{Measurement scenario}

To the covering of measurements and the possible events, we give the name measurement scenario \cite{Abramsky_2018}.

\begin{definition}
A measurement scenario $\left<X,\mathcal{U},(O_{x})_{x\in X}\right>$ is a hypergraph\footnote{Usually one imposes the hypergraph has some additional structure, usually enough to identify it as a simplicial complex. See Ref. \cite{montanhano2021contextuality} for a justified construction of the measurement scenario.} $\left<X,\mathcal{U}\right>$, where $X$ is the set of measurements and $\mathcal{U}$ a covering of contexts (a family of sets of compatible measurements), plus the sets $(O_{x})_{x\in X}$ for each  $x\in X$ are called outcome sets, with their elements the possible events of each measurement.
\end{definition}

For simplicity let's suppose that outcome sets are finite, and therefore one can define an outcome set $O$ for all the measurements $x$\footnote{We can codify any $O_{x}$ through an injective function $f_{x}:O_{x}\to O$; we just need to ignore elements that aren't in the image of $f_{x}$, such that these elements aren't in the distribution's support.}. We will also work with a measurement scenario with a simplicial complex structure of contexts.

\subsection{Presheaves and sheaves}

\begin{definition}
A presheaf is a functor $F:C^{op}\to \textbf{Set}$ of a category $C$ to the category of sets. 
\end{definition}

Let $(C,J)$ be a site, a small category $C$ equipped with a coverage $J$. In other words, any object $U\in C$ admits a collection of families of morphisms $\{f_{i}:U_{i}\to U\}_{i\in I}$ called covering families. 

\begin{definition}
A presheaf on $(C,J)$ is a sheaf if it satisfies the following axioms
\begin{itemize}
    \item Gluing: if for all $i\in I$ we have $s_{i}\in F(U_{i})$ such that $s_{i}|_{U_{i}\cap U_{j}}=s_{j}|_{U_{i}\cap U_{j}}$, then there is $s\in F(U)$ satisfying $s_{i}=s|_{U_{i}}$;
    \item Locality: if $s,t\in F(U)$ such that $s|_{U_{i}}=t|_{U_{i}}$ for all $U_{i}$, then $s=t$.
\end{itemize}
\end{definition}

\begin{definition}
Elements $s\in F(U)$ of the image of a presheaf are called local sections if $U\neq X$, and global sections if $U=X$.
\end{definition}

\begin{definition}
A compatible family is a family of sections $\left\{s_{i}\in F(U_{i})\right\}_{i\in I}$ such that for all $j,k\in I$ holds
$s_{j}|_{U_{j}\cap U_{k}}=s_{k}|_{U_{j}\cap U_{k}}$ in $F(U_{j}\cap U_{k})$.
\end{definition}

\subsection{Sheaf of events}

The covering $\mathcal{U}$ can be restricted to maximal contexts, which will also be denoted by $\mathcal{U}$, so that $\left<X,\mathcal{U}\right>$ can be understood as a set $X$ with a covering $\mathcal{U}$ of maximal contexts $U_{j\in I}$ indexed by an ordered set $I$\footnote{Given a covering, one can construct a locale, a pointless space, using unions and intersections. This means that the measurements are not the fundamental objects, but rather the minimal contexts become the effective measurements of the scenario, depending on how one chooses the covering. A physical example of refinement is spin degeneration, where refinement occurs by applying a suitable magnetic field.}. Since the intersection of contexts is a context, we can define the inclusion morphism $\rho(jk,j): U_{j}\cap U_{k}\to U_{j}$, which turns the set of contexts and the inclusion morphisms into a small category\footnote{From Ref. \cite{Johnstone:592033}, we can see that the category of contexts with the inclusion is a site.}.

\begin{definition}
The outcome sets are defined by a functor $\mathcal{E}:\left<X,\mathcal{U}\right>^{op}\to\textbf{Set}$, with $\mathcal{E}::U\mapsto O^{U}=\bigtimes_{x\in U}O_{x}$ and $\mathcal{E}::\rho\mapsto\rho'$, such that for each element $U\in\mathcal{U}$ we have an outcome set $O^{U}$ of the context and $\rho'$ is the restriction to the outcome sets, $\rho'(j,kj):O^{j}\to O^{kj}=\mathcal{E}(U_{j}\cap U_{k})::s_{j}\mapsto s_{j}|_{kj}$.
\end{definition}

\begin{proposition}
The functor $\mathcal{E}$ is a sheaf in the site of measurements and contexts, called the sheaf of events of a given measurement scenario.
\end{proposition}

\subsection{Empirical model}

$R$-empirical models are defined with a semiring $R$, such as the Boolean semiring $\mathbb{B}$, the reals $\mathbb{R}$, or the probability semiring $\mathbb{R}^{+}$. The choice of an $R$ defines a way to probe the model.

To define $R$-empirical models, we use another functor $\mathcal{D}_{R}:\textbf{Set}\to\textbf{Set}::O^{U}\mapsto\left\{\mu^{O^{U}}_{R}\right\}$, taking a set of local events to the set of $R$-distributions defined on it $\mu^{O^{U}}_{R}:\mathbb{P}\left(O^{U}\right)\to R$ that satisfies $\mu^{O^{U}}_{R}(O^{U})=1_{R}$, in analogy with probabilistic distributions. We will denote by $\mu_{R}::U\in\mathcal{U}\mapsto\mu_{R}^{O^{U}}$ a set of $R$-distribution defined in each element of $\mathcal{U}$, and call it a state. In the morphisms, $\mathcal{D}_{R}::\rho'(j,kj)\mapsto \rho''(j,kj)$, with $\rho''(j,kj)::\mu^{O^{j}}_{R}\mapsto\mu^{O^{j}|_{kj}}_{R}=\mu^{O^{j}}_{R}|_{kj}$ the marginalization of the $R$-distribution $j$ on the intersection $kj$. 

\begin{definition}
The tuple $(X,\mathcal{U},\mathcal{E},\mu_{R})=e_{R}$ is called an $R$-empirical model over the measurement scenario $\left<X,\mathcal{U},(O_{x})_{x\in X}\right>=\left(\left<X,\mathcal{U}\right>,\mathcal{E}\right)$ given by the state $\mu_{R}$, defining a set of local sections $\left\{\mu_{R}^{O^{U}}\in\mathcal{D}_{R}\mathcal{E}(U);U\in\mathcal{U}\right\}$.
\end{definition}

\subsection{No-disturbance}

The no-disturbance condition is a usual condition imposed in an $R$-empirical model, sometimes implicitly. It says that $\mu^{O^{j}}_{R}|_{kj}=\mu^{O^{k}}_{R}|_{kj}$ for all $k$ and $j$, which means there is local agreement between contexts. This condition is equivalent to the existence of a compatible family to $\mathcal{D}_{R}\mathcal{E}$, but it doesn't imply in $\mathcal{D}_{R}\mathcal{E}$ to be a sheaf. Since we can only have access to contexts, it is possible to define the functor $\mathcal{D}_{R}\mathcal{E}$ through a state that can't be extended to a distribution in the global events.

No-disturbance is equivalent to the notion of parameter-independence, as explained in Ref. \cite{barbosa2019continuousvariable}, a property that, if violated, means the existence of non-trivial data between contexts. As stated in Ref. \cite{Dzhafarov_2018}, where disturbance is called inconsistent connectedness: ``Intuitively, inconsistent connectedness is a manifestation of direct causal action of experimental set-up upon the variables measured in it''. We will work with non-disturbing models.

\subsection{Contextuality}

Contextuality is the impossibility of describing a given $R$-empirical model in classical terms, but one must first define which classical notion to use. We will call it $R$-contextuality to make explicit the chosen semiring. First, we know that any distribution can be described as the marginalization of another one,
\begin{equation}
    \mu_{R}^{O^{U}}(A)=\sum_{\lambda\in\Lambda}k^{O^{U}}\left(\lambda,A\right),
\end{equation}
for all $A\in\mathbb{P}(O^{U})$, where $k^{O^{U}}:\Lambda\times\mathbb{P}(O^{U})\to R$ is an $R$-distribution that satisfies $\sum_{\lambda\in\Lambda}k^{O^{U}}\left(\lambda,O^{U}\right)=1_{R}$. In the literature of contextuality and nonlocality, $\Lambda$ is called the set of hidden variables, which is statistically taken into account but is empirically inaccessible. 

To impose a classical behavior, the hidden variables must be independent of the contexts, a property called lambda-independence\footnote{Lambda-independence is related to the concept of free choice in nonlocality \cite{Cavalcanti_2018,Abramsky_2014}. It can be understood as a dependence of the hidden variables, sometimes called ontic variables, on the contexts. Such dependence can store contextuality, as free choice can be understood as storing nonlocality \cite{Blasiak_2021}. For more details on the classification of hidden variables in the subject of nonlocality, see Ref. \cite{Brandenburger_2008}.}. To reflect such behavior of independence, our model must show independence between measurements, in other words, be factorizable. Such independence allows us to write
\begin{equation}
    \mu_{R}^{O^{U}}(A)=\sum_{\Lambda}p\left(\lambda,A\right)\prod_{x\in U}\mu_{R}^{O^{x}}(\rho'(U,x)(A)),
\end{equation}
with the assistance of the set of hidden variables $\Lambda$ being statistically taken into account by a distribution $p:\Lambda\times\mathbb{P}(O^{U})\to R$. Summing it up with lambda-independence implies
\begin{equation}
    \mu_{R}^{O^{U}}(A)=\sum_{\Lambda}p\left(\lambda\right)\prod_{x\in U}\mu_{R}^{O^{x}}(\rho'(U,x)(A)),
\label{Contextualidade}
\end{equation}
with $p(\Lambda)=1_{R}$, closing the representation of a $R$-empirical model as a classical system.

\begin{definition}
An $R$-empirical model is said to be $R$-noncontextual if there is an $R$-distribution $p$ and a set of hidden variables $\Lambda$ such that equation \ref{Contextualidade} holds for all $U\in\mathcal{U}$.
\end{definition}

Another property we can impose is called outcome-determinism, which is the property of logically distinguish between outcomes.

\begin{definition}
Outcome-determinism in an $R$-distribution of an empirical model is defined as for all $\lambda\in\Lambda$ there is an outcome $o\in O^{U}$ such that $k^{O^{U}}(\lambda,A)=\delta_{o}(A)$. Equivalently, an empirical model is outcome-deterministic if it satisfies $\prod_{x\in U}\mu_{R}^{O^{x}}(\rho'(U,x)(A))\in\{0,1\}$. 
\end{definition}

In combination with no-disturbance, we get the following result \cite{Abramsky2011,Abramsky_2017}.

\begin{proposition}
A non-disturbing $R$-empirical model that satisfies the outcome-determinism condition has as its hidden variables exactly its global events.
\end{proposition}

This result allows us to develop a measure of contextuality called 'contextual fraction' through the use of linear programming \cite{Abramsky_2017,barbosa2019continuousvariable}, once we know the set $\Lambda$. With it one can also prove the Fine–Abramsky–Brandenburger Theorem \cite{Abramsky2011}, where $R$-contextuality can be understood as the non-extendability of a local section to a global section of $\mathcal{D}_{R}\mathcal{E}$, or in other words, as the nonexistence of a global $R$-distribution with marginalization to a context $U\in\mathcal{U}$. 

\begin{theorem}[Fine-Abramsky-Brandenburger]
For an empirical model satisfying no-disturbance and outcome-determinism, the following are equivalent:
\begin{itemize}
    \item to be described by deterministic hidden variables described by equation \ref{Contextualidade};
    \item all local sections extending to global sections;
    \item a distribution $\mu^{O^X}_R$ that marginalizes to $\mu^{O^{U}}_R$.
\end{itemize}
\label{FAB}
\end{theorem}

We can graphically describe noncontextual behavior as the commutation of the diagram
\begin{equation}
    \begin{tikzcd}
\mathcal{E}(\mathcal{U}) \arrow{rr}{\mu_{R}} \arrow[hookrightarrow]{dr}{i'} & & R \\
& \mathcal{E}(X) \arrow{ur}{\nu_{R}} &
\end{tikzcd}
\end{equation}
The global events define a global $R$-distribution $\mu_{R}^{O^{X}}$, and the commutation implies the realization of the $R$-empirical model by it. Here, $i'$ is the inclusion of local events in global events. As explored in Ref. \cite{https://doi.org/10.48550/arxiv.2202.08719}, the failure of commutativity can be seen in two independent ways: the first due to $i'$, which is linked to anti-realist interpretations, and the second due to $\nu_{R}$, linked to realist interpretations. With the previous results, the Sheaf Approach chooses to attribute the failure to $\nu_R$ by imposing the sheaf condition on events.

\subsection{A hierarchy of contextuality}

The Sheaf Approach generally restricts itself to models satisfying outcome-determinism and non-perturbation to be able to utilize the strong prior results. Under such conditions, we will cite some facts that follow from the choice of a semi-ring $R$.

Every noncontextual non-disturbing $\mathbb{R}$-empirical model is equivalent to a non-disturbing $[0,1]$-empirical model. If the $\mathbb{R}$-empirical model allows a non-negative description, then it is $[0,1]$-noncontextual. $[0,1]$-contextuality is referred to as probabilistic contextuality. An example is the Bell-Clauser-Horne-Shimony-Holt model \cite{bell_aspect_2004,PhysRevLett.23.880}, which exhibits contextuality only when dealing with probabilities since it does not feature any local event that does not allow an extension to a global event. In other words, contextuality is only observable through probability distributions and not logically. 

$\mathbb{B}$-contextuality is called logical or possibilistic. A $\mathbb{B}$-contextual empirical model has at least one local event that does not allow extension to a global event, and therefore the model allows verification of non-classicality through logical methods. An example of this case is the Hardy model \cite{Phys,Rev,Cabello_2013}. If the $\mathbb{B}$-contextual empirical model does not exhibit any local event that can be extended to a global event, then it is said to exhibit strong contextuality, the highest level of contextuality a model can exhibit. Examples of this case include the Popescu-Rohrlich boxes model \cite{Popescu1994}, the Klyachko-Can-Binicioğlu-Shumovsky pentagram model \cite{Lett}, and the Greenberger-Horne-Zeilinger model \cite{greenberger2007going}. Every $[0,1]$-empirical model can be decomposed into a noncontextual part and a strong contextual part.

A $\mathbb{B}$-empirical model uniquely defines a $[0,1]$-model induced by the standard map $[0,1]\to \mathbb{B}$ defined by $0\mapsto 0$ and $(0,1]\mapsto 1_{\mathbb{B}}$. This determines that logical contextuality is stronger than probabilistic contextuality, as there exist models that do not exhibit logical contextuality even when they are contextual. On the other hand, strong contextuality is a special case of logical contextuality, since no local event admits an extension. We then have a hierarchy relating the contextuality of non-perturbative models:
\begin{widetext}
\begin{equation}
\text{strongly contextual} > \text{logical contextual} > \text{probabilistic contextual} > \text{noncontextual}
\end{equation}
\end{widetext}

\section{Proof of Proposition \ref{Prop}}
\label{proof}

For entirety, we will repeat the statement of the proposition.

\begin{proposition}
The following statements are valid:
\begin{itemize}
    \item Axiom $\mathbf{T}$ turns trust relations vacuous.
    \item The trust relation $\stackrqarrow{D}{}$, along with the condition that $(\phi\to D_I\phi)\forall\phi$, induces a fundamental truth.
\end{itemize}
 
\end{proposition}

\begin{proof}
    According to Axiom $\mathbf{T}$, we have $(K_j\phi\to\phi)\forall\phi$, which implies $(K_i(K_j\phi)\to K_i\phi)\forall\phi$. In other words, the existence of fundamental truth makes the knowledge of every agent reliable. This allows us to generalize to sets of agents and trust relations given by $D_G$ and $E_G$ by means of $(D_I\phi\to\phi)\forall\phi$ and $(E_I\phi\to\phi)\forall\phi$, which follow from their definitions. Hence, we have $(K_i(D_I\phi)\to K_i\phi)\forall\phi$ and $(K_i(E_I\phi)\to K_i\phi)\forall\phi$, meaning they are also rendered reliable. Therefore, we obtain vacuity.

    Suppose a set $I$ of agents and their trust relation $\stackrqarrow{D}{}$. The construction of $D_I$, even more explicitly in its topology $\tau_{D_I}$, depends on the trust relation $\stackrqarrow{D}{}$ between sets of agents. By imposing that $(\phi\to D_I\phi)\forall\phi$, we are saying that $(\phi\leftrightarrow D_I\phi)\forall\phi$, that the distributed knowledge is equivalent to a fundamental truth. As it has the finest topology, there is no other knowledge that the agents can access beyond what is captured by $D_I$. Furthermore, $(K_i D_I\phi\to K_i\phi)\forall\phi$, meaning all agents trust the distributed knowledge.
\end{proof}

\bibliographystyle{apsrev4-2}
\bibliography{apssamp2.bib}

\end{document}